\documentclass[12pt,onecolumn,draftclsnofoot,journal]{IEEEtran}
\usepackage{graphicx}
\usepackage{float}
\usepackage{epstopdf}
\usepackage[cmex10]{amsmath}
\usepackage{array}
\usepackage{cite}
\usepackage{amssymb}
\usepackage{amsfonts}
\usepackage{amsmath}
\usepackage{stackrel}
\usepackage{booktabs,multirow}
\usepackage{arydshln}
\usepackage{slashbox}
\usepackage{amsthm}
\usepackage{listings}
\usepackage{algorithm}
\usepackage{xcolor}
\usepackage{algorithmicx}
\usepackage{algpseudocode}
\usepackage{threeparttable}

\newtheorem{lem}{Lemma}
\newtheorem{rem}{Remark}
\newtheorem{theo}{Theorem}

\newtheorem{cor}{Corollary}

\newtheorem{fact}{Fact}
\newtheorem{ex}{Example}
\newfloat{routine}{htbp}{loa}
\floatname{routine}{Routine}

\makeatletter
\newcommand{\algmargin}{\the\ALG@thistlm}
\makeatother
\newlength{\forwidth}
\algdef{SE}[parFOR]{parFor}{EndparFor}[1]
  {\parbox[t]{\dimexpr\linewidth-\algmargin}{
     \hangindent\forwidth\strut\algorithmicfor\ #1\ \algorithmicdo\strut}}{\algorithmicend\ \algorithmicfor}
\algnewcommand{\parState}[1]{\State
  \parbox[t]{\dimexpr\linewidth-\algmargin}{\strut #1\strut}}

\newlength{\ifwidth}
\algdef{SE}[parIF]{parIf}{EndparIf}[1]
  {\parbox[t]{\dimexpr\linewidth-\algmargin}{
     \hangindent\ifwidth\strut\algorithmicif\ #1\ \algorithmicdo\strut}}{\algorithmicend\ \algorithmicif}

\begin{document}

\title{From Cages to Trapping Sets and Codewords: A Technique to Derive Tight Upper Bounds on the Minimum Size of Trapping Sets and Minimum Distance of LDPC Codes}
\author{Ali Dehghan, and Amir H. Banihashemi,\IEEEmembership{ Senior Member, IEEE}}

\maketitle


\begin{abstract}
Cages, defined as regular graphs with minimum number of nodes for a given girth, are well-studied in graph theory. Trapping sets are graphical structures responsible for error floor of low-density parity-check (LDPC) codes, and are well investigated in coding theory. In this paper, we make connections between cages and trapping sets. In particular, starting from a cage (or a modified cage), we construct a trapping set in multiple steps. Based on the connection between cages and trapping sets, we then use the available results in graph theory on cages and derive tight upper bounds on the size of the smallest trapping sets for variable-regular LDPC codes with a given variable degree and girth. The derived upper bounds in many cases meet the best known lower bounds and thus provide the actual size of the smallest trapping sets. Considering that non-zero codewords are a special case of trapping sets, we also derive tight upper bounds on the minimum weight of such codewords, i.e., the minimum distance, of variable-regular LDPC codes as a function of variable degree and girth.

\begin{flushleft}
\end{flushleft}

\end{abstract}

\section{introduction}

Trapping sets are combinatorial objects known to be responsible for the error floor of low-density parity-check (LDPC) codes. They can be viewed as subgraphs of the code's Tanner graph induced by a certain number $a$ of variable nodes. If the number of odd-degree check nodes in the subgraph is $b$, the trapping set is said to belong to the $(a,b)$ class.
Among trapping sets, those with only degree-$1$ and degree-$2$ check nodes are known to be the most harmful. Such trapping sets are called {\em elementary}. Within the category of elementary trapping sets (ETSs), those in which each variable node is connected to at least two even-degree check nodes are referred to as {\em leafless}. Leafless ETSs (LETSs) appear often as the most dominating in the error floor. Elementary trapping sets and LETSs have been the subject of extensive research~\cite{MSW, ZR, asvadi2011lowering, MR2991821, MR3252383, hashemi2015characterization2, hashemi2015new, Y-Arxiv, hashemilower, hashemi2018}. The majority of existing literature is focused on the characterization and efficient search of trapping sets, see, e.g.,~\cite{hashemi2015new}, and the references therein.

In general, for a given $b$, trapping sets with smaller $a$ are deemed to be more harmful. Recently, in~\cite{hashemilower}, lower bounds on the smallest size of ETSs and non-elementary trapping sets (NETSs) of variable-regular LDPC codes in terms of variable degree $d_v$ and girth $g$ were derived. It was shown in~\cite{hashemilower}, that the minimum size of NETSs is, in general, larger than that of ETSs for a given $b$ value. This provided a theoretical justification for why NETSs are in general less harmful than ETSs.

An $(a,b)$ trapping set with $b=0$ is a non-zero codeword of weight $a$. Based on this view, most recently, in~\cite{hashemi2018}, non-zero codewords were partitioned into two categories of elementary and non-elementary, and lower bounds on the minimum weight of codewords in each category were derived. These bounds along with an efficient search algorithm~\cite{hashemi2015new}, \cite{Y-Arxiv} were then used in~\cite{hashemi2018} to derive tight lower and upper bounds on the minimum distance of specific variable-regular LDPC codes. In particular, the upper bound in~\cite{hashemi2018} was derived by searching for an elementary codeword within the code's Tanner graph. In related work, Vasic and Milenkovic \cite{MR2094873} derived bounds on the minimum distance of LDPC codes constructed based on balanced incomplete block designs (BIBDs). In~\cite{rr1}, Smarandache and Vontobel derived upper bounds on the minimum distance of QC-LDPC codes. Rosnes {\em et al.}~\cite{rr2} improved upper bounds on the minimum distance of array LDPC codes. Compared to the upper bounds of \cite{MR2094873, rr1, rr2}, which apply to specific categories of structured LDPC codes, the upper bounds derived here are more general in the sense that they are applicable to any variable-regular LDPC code regardless of whether such a code is structured or randomly constructed.

In this work, we derive tight upper bounds on the minimum size $a$ of $(a,b)$ trapping sets with a given $b$ within a variable-regular Tanner graph with a given $d_v$ and $g$. The bounds are specific to different categories of trapping sets including ETSs, LETSs, NETSs and ETSs with leafs (ETSLs). For $b=0$, the bounds result in upper bounds on the minimum distance of LDPC codes. The derived bounds in many cases are equal to the best known lower bounds and are thus tight. To the best of our knowledge, the derived upper bounds are the first of their kind in that they relate the minimum size of trapping sets or the smallest minimum distance of LDPC codes to the girth and the degree distribution of the code. While, it has been long established that minimum distance of an ensemble of regular LDPC codes increases linearly with the block length $n$ on average, see, e.g.,~\cite{LS}, our results concern finite-length LDPC codes, and demonstrate that, in the worst case, the minimum distance of an ensemble is only a function of $d_v$ and $g$ and is independent of $n$ (constant in $n$). In fact, an important aspect of the results presented in this work is to identify structures of trapping sets and non-zero codewords of the smallest possible size so that they can be found and/or avoided in code constructions.  

In this work, we focus on the category of variable-regular LDPC codes. For this category, it is well established that LETSs are the main culprits in the error floor region~\cite{hashemi2015new}. Simulation results however, show that while 
LETSs well identify the initial position of variable nodes that can cause a decoder failure in the error floor region, the eventual positions in which the errors will occur are in some cases ETSLs. This makes the study of ETSLs worthwhile. We also note that while NETSs appear rarely as errors in the error floor region of many existing variable-regular LDPC codes, this is, at least partly, due to the fact that in many of the existing codes, there are smaller ETSs in the code's Tanner graph that dominate the error floor performance. Should one design new codes that are free of smaller ETSs such that the dominant classes contain both ETS and NETS structures, the relative harmfulness of NETSs can no longer be ignored. This motivates the results presented in this work on NETSs.

An important aspect of the results presented in this paper is the novelty of the technique used to obtain the bounds. We make connections between {\em cages}, which are well-studied in graph theory, and trapping sets. In particular, for each category of trapping sets (ETSs, LETSs, NETSs, ETSLs, non-zero codewords), we start from a cage (or a slightly modified cage) and construct, in multiple steps, a trapping set within the category of interest. The existing results on the size of the cages then allow us to establish upper bounds on the size of trapping sets.

A cage is an $r$-regular graph that has minimum number $n(r,g)$ of nodes for a given girth $g$.
Finding cages for different and increasingly larger values of $r$ and $g$ has been an active and technically rich area of research for decades, see, e.g.,~\cite{MR0211917, MR642638, MR771732, MR1377553, MR2375518, MR2601271, MR2773587, exoo2008dynamic, MR3127005, MR3131381, MR3441655, cage}.
There are a variety of techniques used in studying cages. In general, the following three main categories can be identified: ($1$) search-based techniques, see, e.g.,~\cite{MR771732},~\cite{MR1377553},~\cite{MR2773587},~\cite{exoo2008dynamic}; ($2$) excision methods, see, e.g.,~\cite{MR2601271},~\cite{MR3441655}; and ($3$) geometric techniques, see, e.g.,~\cite{MR0211917},~\cite{MR2375518},~\cite{exoo2008dynamic}. While, a detailed review of these techniques is beyond the scope of this paper, we believe, through the connections established in this work, an interested reader would be able to apply similar techniques to further study the trapping sets of LDPC codes.

The organization of the rest of the paper is as follows: In Section~\ref{sec2}, we provide some definitions, notations and review the best known lower bounds on the smallest size of trapping sets.
In this section, we also present some background on cages. Section~\ref{sec2.5} contains the main contributions of this paper. We start the section by explaining the general technique used for the derivation of the upper bounds and go on
to derive the bounds for different categories of trapping sets, i.e., LETSs, ETSLs and NETSs.
Finally, the paper is concluded with some remarks in Section~\ref{sec6}.

\section{Preliminaries}
\label{sec2}

\subsection{Definitions and Notations}

For a graph $G$, we denote the node set and the edge set of $G$ by $V(G)$ and $E(G)$, or $V$ and $E$, if there is no ambiguity, respectively. In this work, we consider undirected graphs with no loop or parallel edges.
For $v \in V(G)$ and $S \subseteq V(G)$, notations $N(v)$ and $N(S)$ are used to denote
the neighbor set of $v$, and the set of nodes of $G$ which has a neighbor in $S$, respectively.
A {\it path} of length $c$ in the graph $G$ is a sequence of distinct nodes $v_1, v_2, \ldots , v_{c+1}$ in $V(G)$, such that $\{v_i, v_{i+1}\} \in E(G)$, for $1 \leq i \leq c$.
A {\it cycle} of length $c$ is a sequence of distinct  nodes
$v_1, v_2, \ldots , v_{c}$ in $V(G)$ such that $v_1, v_2, \ldots , v_{c}$ form a path of length $c-1$, and
$\{v_c, v_{1}\} \in E(G)$. We may refer to a path or a cycle by the set of their nodes or by the set of their edges. The length of the shortest cycle(s) in a graph is called the {\em girth} of the graph and is denoted by $g$.

A graph $G=(V,E)$ is called {\it bipartite}, if the node set $V(G)$ can be partitioned into two disjoint subsets $U$ and $W$ (i.e., $V(G) = U \cup W \text{ and } U \cap W =\emptyset $), such that every edge in $E$ connects a node
from $U$ to a node from $W$. The {\it Tanner graph} of a low-density parity-check (LDPC) code is a bipartite graph, in
which $U$ and $W$ are referred to as {\it variable nodes} and {\it check nodes}, respectively.

The number of edges incident to a node $v$ is called the {\em degree} of $v$, and is denoted by $d(v)$.
A bipartite graph $G = (U\cup W,E)$ is called {\it bi-regular}, if all the nodes on the same side of the bipartition have the same degree,
i.e., if all the nodes in $U$ have the same degree $d_u$, and all the nodes in $W$ have the same degree $d_w$.
It is clear that, for a bi-regular graph, $|U|d_u=|W|d_w=|E(G)|$. A bipartite graph that is not bi-regular is called {\it irregular}.
A Tanner graph $G = (U\cup W,E)$ is called {\em variable-regular} with variable degree $d_v$,
if for each variable node $u_i\in U$, $d(u_i) = d_v$. Also, a $(d_v, d_c)$-regular Tanner graph is a variable-regular graph with variable degree $d_v$, in which all check nodes have the same degree $d_c$.

In a Tanner graph $G = (U\cup W,E)$, for a subset $S$ of $U$, the induced subgraph of $S$ in the graph $G$, denoted by $G(S)$, is the graph with the set of nodes
$S\cup N(S)$, and the set of edges $\{\{u_i,w_j\}: \{u_i,w_j\} \in E(G) , u_i\in S , w_j \in N(S)\}$.
The set of check nodes with odd and even degrees in $G(S)$ are denoted by $N_o(S)$ and $N_e(S)$, respectively.
Also, the terms {\em unsatisfied} check nodes and {\em satisfied} check nodes are used to refer to the check nodes in $N_o(S)$ and $N_e(S)$, respectively.
Throughout this paper, the size of an induced subgraph $G(S)$ is defined to be the number of its variable nodes (i.e., $|S|$).

For a given Tanner graph $G$, a set $S \subset U$, is said to be an {\it $(a,b)$ trapping set (TS)} if $|S| = a$ and
$|N_o(S)| = b$. Alternatively, set $S$ is said to belong to the {\em class} of $(a,b)$ TSs. An {\it elementary trapping set
(ETS)} is a TS for which all the check nodes in $G(S)$ have degree one or two.
A {\it leafless ETS (LETS)} $S$ is an ETS for which each variable node in $S$ is connected to at least two satisfied check nodes in $G(S)$. Otherwise, the set $S$ is called an ETS with leaf (ETSL).
A {\em non-elementary trapping set (NETS)} is a trapping set which is not elementary. Trapping sets are thus partitioned into three categories of LETSs, ETSLs and NETSs.
We use the notations $\mathcal{T}^{b}_{LETS}(d_v;g)$,  $\mathcal{T}^{b}_{ETSL}(d_v;g)$ and $\mathcal{T}^{b}_{NETS}(d_v;g)$ to denote the smallest size of LETSs, ETSLs and NETSs with $b$ unsatisfied check nodes in variable-regular LDPC codes with variable degree $d_v$ and girth $g$, respectively.

The {\em normal graph} of an ETS $S$ is obtained from $G(S)$ by removing all the check nodes of degree one and their incident edges, and by replacing all
the degree-$2$ check nodes and their two incident edges by a single edge~\cite{MR2991821}. The normal graph of a LETS has no node with degree one, hence, the terminology ``leafless''~\cite{hashemi2015new}.
The normal graphs of ETSLs, on the other hand, have at least one node with degree one.
Since in this work, we consider Tanner graphs with no parallel edges, the girth $g$ of the graphs is at least $4$. The normal graph of an ETS in a Tanner graph with $g=4$ can have parallel edges. Low-density parity-check codes whose Tanner graphs have $g=4$, however, perform poorly under iterative message-passing decoding algorithms. In this paper, we thus consider Tanner graphs with $g \geq 6$. The normal graphs of ETSs in such Tanner graphs have no parallel edges.

The {\em minimum distance} of a linear block code, and thus an LDPC code, is the minimum weight of its non-zero codewords. A non-zero codeword of weight
$a$ in an LDPC code can be viewed as an $(a,0)$ trapping set in the code's Tanner graph. The non-zero codewords of an LDPC code can be partitioned into
two categories of {\em elementary} and {\em non-elementary}, where in the first category, all the satisfied check nodes in the subgraph have degree $2$.
Consider an LDPC code with minimum variable degree at least two. Any elementary (non-elementary) codeword of such an LDPC code corresponds to a LETS (NETS) in the code's Tanner graph.

\subsection{Best known lower bounds on the minimum size of trapping sets}

The following theorem provides the best known general lower bound on the minimum size of trapping sets of variable-regular LDPC codes.

\begin{theo}
\label{lowf}
\cite{hashemilower} Consider a variable-regular Tanner graph $G$ with variable degree $d_{v}$ and girth  $g$.
A lower bound on the size $a$ of an $(a,b)$ trapping set in the graph $G$, whose induced subgraph contains a check node of degree $k\:(\geq 2)$ is given in (\ref{ineq1}),
where $b'=b-(k \mod 2)$, $T=k(d_{v}-1)-b'$, and $b$ is assumed to satisfy $b < k(d_{v}-1)+(k \mod 2)$. (Notation $\mod$ is for modulo operation.)
\begin{table*}[h]
\begin{align}
a \geq
\left\{
\begin{array}{ll}
k+ T \sum\limits_{i=0}^{\lfloor g/4 -2  \rfloor}  (d_{v}-1)^i,    & \text{for}~g/2~\text{even}, \\
k+T \sum\limits_{i=0}^{\lfloor g/4 -2  \rfloor} (d_{v}-1)^i + \max \{\lceil (T (d_{v}-1)^{\lfloor g/4-1\rfloor})/d_{v} \rceil, (d_{v}-1-\lfloor b'/ k \rfloor)(d_{v}-1)^{\lfloor g/4-1\rfloor}\},   & \text{for}~g/2~ \text{odd},
\end{array}
\right.
\label{ineq1}
\end{align}
\end{table*}
\end{theo}

Note that Theorem \ref{lowf} with $k=2$ provides a lower bound on the size of the smallest possible ETSs. Also, Theorem \ref{lowf} with $k=3$ ($k=4$) provides
a lower bound on the size of the smallest possible NETSs with $b > 0$ ($b=0$). The lower bounds of Theorem~\ref{lowf} in some cases are improved in~\cite{hashemilower} for LETSs and ETSLs based on the $dpl$ characterizations of~\cite{hashemi2015new} and~\cite{Y-Arxiv}. Recall that the case of $b=0$ for LETSs and NETSs corresponds to non-zero elementary and non-elementary codewords, respectively.
In particular, for the minimum weight of  non-elementary codewords of a variable-regular LDPC code with
variable degree $d_v$ and girth $g$, the result of Theorem~\ref{lowf} reduces to the following lower bound \cite{hashemi2018}:

\small
\begin{equation}
L_{ne} = \left\{
\begin{array}{lr}
 \max \{\lceil 4 (d_{v}-1)^{\lfloor g/4\rfloor}/d_{v} \rceil, (d_{v}-1)^{\lfloor g/4\rfloor}\} +
\sum\limits_{i=0}^{\lfloor g/4 -1  \rfloor} 4 (d_{v}-1)^i , &  \hspace{2cm} \text{for}~g/2~ \text{odd},
\\
 \sum\limits_{i=0}^{\lfloor g/4 -1  \rfloor} 4 (d_{v}-1)^i,  &  \hspace{2cm} \text{for}~g/2~\text{even}.
\end{array}
\right.
\label{ineq2}
\end{equation}
\normalsize

\subsection{Background on cages}
\label{subseczz}
An {\it $(r;g)$-graph} is defined to be a graph with girth $g$ in which the degree of each node is $r$. It is well-known that for any $r \geq 2$ and $g \geq 3$, an $(r;g)$-graph exists~\cite{exoo2008dynamic}.
An {\it $(r;g)$-cage} is an $(r;g)$-graph with the fewest possible number of nodes. The number of nodes in an $(r;g)$-cage is denoted by $n(r;g)$. In the rest of the paper, we sometimes refer to $n(r;g)$ as the {\em size} of the corresponding cage. It is known~\cite{exoo2008dynamic} that
\begin{equation}\label{Moore}
n(r;g)\geq \begin{cases}
1+ r \sum_{i=0}^{(g-3)/2} (r-1)^i       & \text{if } g \text{ is odd},\\
2 \sum_{i=0}^{g/2-1}  (r-1)^i            & \text{if } g \text{ is even}.\
\end{cases}
\end{equation}
The above lower bound is called the {\it Moore bound}, and graphs for which equality holds are called {\it Moore graphs} \cite{exoo2008dynamic}.
Table \ref{Table1} summarizes the known values of $n(r;g)$ for small values of $r$ and $g$~\cite{exoo2008dynamic}.

\begin{table}[ht]
\caption{Known values of $n(r;g)$ for small values of $r$ and $g$ \cite{exoo2008dynamic}.}
\begin{center}
\scalebox{1}{
\begin{tabular}{ |c||c|c|c|c|c|c|c|c| c|c| }
\hline
\backslashbox{$r$}{$ g $}      &  3  & 4    & 5     &6     & 7       &8     & 9     & 10  & 11     & 12   \\
\hline
\hline
3                              & 4   &6     & 10    & 14   & 24      &30    & 58     &70  & 112     & 126   \\
\hline
4                              & 5   &8     & 19    &26    &67       &80    & 275    &384 & --      & 728 \\
\hline
5                              & 6   &10    & 30    &42    &152      &170   & --     &--   & --     &2730  \\
\hline
6                              & 7   &12    &40     &62    &294      &312   & --     &--   & --     &7812 \\
\hline
7                              & 8   &14    & 50    &90    &--       &--    & --     &--   & --     &--  \\
\hline
\end{tabular}
}
\end{center}
\label{Table1}
\end{table}

In this work, we also use a generalization of cages to derive some of the upper bounds. Let $r,s,g$ be positive integers such that $r \geq 3$, $g \geq 3$ and $r < s$.
An {\it $(r,s;g)$-graph} is defined to be a graph with girth $g$, in which each node has degree $r$ or $s$. An {\it $(r,s;g)$-cage} is defined as an $(r,s;g)$-graph with the fewest possible number of nodes.
The number of nodes in an $(r,s;g)$-cage is denoted by $n(r,s;g)$. In this work, we are interested in  $(r,r+1;g)$-graphs with only one node of degree $r+1$.
We thus define an {\it $(r,s;g)$-good-graph} to be an $(r,s;g)$-graph that has only one node of degree $s$.
For some values of $r$, $s$ and $g$, there is no $(r,s;g)$-good-graph. For instance, if $r$ is an even number and $s$ is an odd number, then there is no $(r,s;g)$-good-graph.
For each triple $(r,s,g)$, if there is at least one $(r,s;g)$-good-graph, then an {\it $(r,s;g)$-good-cage} is defined to be
an $(r,s;g)$-good-graph with the fewest possible number of nodes. Notation $n'(r,s;g)$ is used to denote the number of nodes in an $(r,s;g)$-good-cage. If such a cage does not exist, we use $n'(r,s;g)=\infty$.

\section{From cages to Trapping Sets: Derivation of Upper Bounds}
\label{sec2.5}

The main idea involved in creating a trapping set from a cage is {\em graph subdivision}.
For a given graph $G$, a {\em subdivision} of an  edge $e$ with endpoints $\{u,v\}$ yields a graph containing one new node $w$, and with an edge set in which $e$ is replaced by two new edges: $\{u,w\}$ and $\{w,v\}$.
A {\it subdivision} of a graph $G$, denoted by $G^{\frac{1}{2}}$, is a graph resulting from the subdivision of every edge in $G$. In $G^{\frac{1}{2}}$, we can partition the nodes into two sets $V_{old}(G^{\frac{1}{2}})$ and $V_{new}(G^{\frac{1}{2}})$, where $V_{old}(G^{\frac{1}{2}})=V(G)$. To create a bipartite (Tanner) graph from the graph $G^{\frac{1}{2}}$, we replace each node in $V_{old}$ with a variable node and each node in $V_{new}$ with a check node. We denote the resulting graph by $\widetilde{G^{\frac{1}{2}}}$. Fig.~\ref{F99}(b) shows $\widetilde{G^{\frac{1}{2}}}$ for the $(3;3)$-cage $G$, shown in Fig.~\ref{F99}(a).

\begin{figure}[]
	\centering
	\includegraphics [width=0.5\textwidth]{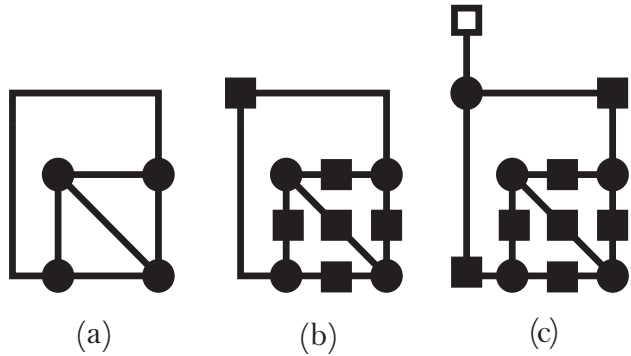}
	\caption{(a) The $(3;3)$-cage $G$, (b) Tanner graph $\widetilde{G^{\frac{1}{2}}}$ (a $(4,0)$ LETS structure in variable-regular LDPC codes with $d_v=3$ and $g=6$), (c) Modification of $\widetilde{G^{\frac{1}{2}}}$ to a $(5,1)$ LETS structure.}
	\label{F99}
\end{figure}

To construct a trapping set with specific $b$, $g$ and $d_v$ values, we start from a suitable cage $G$. Cage $G$ can be a $(d_v;\frac{g}{2})$-cage, a  $(d_v;\frac{g+2}{2})$-cage or a generalized cage, as described in Subsection~\ref{subseczz}. We then
convert the cage $G$ to a Tanner graph $\widetilde{G^{\frac{1}{2}}}$. Finally, by some simple modifications on $\widetilde{G^{\frac{1}{2}}}$ (or multiple copies of $\widetilde{G^{\frac{1}{2}}}$), such as adding or removing some variable nodes, check nodes or edges, we construct a trapping set. In the following, we explain the process and derive the upper bounds for LETSs, ETSLs and NETSs, in Subsections~\ref{sec3} to~\ref{sec5}, respectively. Within each subsection,
we cover different values of $d_v \geq 3$, $b \geq 0$, and $g \geq 6$. We note that cases with $d_v \geq 3$ and $g \geq 6$ are those of interest in the context of variable-regular LDPC codes. In each subsection, we first derive general upper bounds. These bounds are then calculated for specific values of $d_v$, $g$ and $b$, and compared with the best known lower bounds. The range of specific values considered in this paper are $3 \leq d_v \leq 6$, $6 \leq g \leq 16$, and $0 \leq b \leq 5$. To see the theoretical and practical relevance of LDPC codes with larger girths and variable node degrees, the reader is referred to~\cite{MR2951330, MR2400586, MR3480065, MR2236186, MR2245116, eee1, MR3285723, MR2798990, eee2}, and the references therein. In general, LDPC codes with larger girths are desirable as they perform better under iterative message-passing decoding algorithms. Lerger variable node degrees are attractive as they improve the error floor performance. 

It is known that depending on the values of $d_v$ and $b$, certain trapping set classes may not exist. The following lemma describes such cases.

\begin{lem}
\label{lem:cannot}
\cite{hashemilower}
In a variable-regular Tanner graph with variable degree $d_{v}$, (i) if $d_{v}$ is odd, then there does not exist any $(a,b)$ TS with odd $a$ and even $b$, or with even $a$ and odd $b$;
and (ii) if $d_{v}$ is even, then there does not exist any $(a,b)$ TS with odd $b$.
\end{lem}

\subsection{Leafless elementary trapping sets (LETSs)}
\label{sec3}

In the following, we first derive general upper bounds on, or provide the exact value of, $\mathcal{T}^{b}_{LETS}(d_v;g)$ for different values of $g \geq 6$, $d_v \geq 3$ and $b \geq 0$. We then improve some of these bounds for specific values of $g$, $d_v$ or $b$.

\begin{theo}\label{Th92}
(i) For any $g \geq 6$ and any $d_v \geq 3$, we have (a) $\mathcal{T}^{0}_{LETS}(d_v;g)=n(d_v;\frac{g}{2})$, (b) $\mathcal{T}^{d_v-2}_{LETS}(d_v;g) \leq n(d_v;\frac{g}{2})+1$,  (c) $\mathcal{T}^{d_v}_{LETS}(d_v;g) \leq n(d_v;\frac{g}{2})-1$, and (d) $\mathcal{T}^{b}_{LETS}(d_v;g) \leq n(d_v;\frac{g}{2})$, if $b$ is an even value satisfying $b < d_v+2$.\\
(ii) For any $g \geq 8$ and $d_v \geq 3$, we have $\mathcal{T}^{2(d_v-1)}_{LETS}(d_v;g) \leq n(d_v;\frac{g}{2})-2$.\\
(iii) For any $g \geq 6$ and $d_v \geq 4$, or $d_v=3$ and $g \geq 10$, we have $\mathcal{T}^{d_v+2}_{LETS}(d_v;g) \leq n(d_v;\frac{g}{2})-1$.
\end{theo}

\begin{proof}{
$(i)(a)$ Let $H$ be a $(d_v;\frac{g}{2})$-graph. The Tanner graph $\widetilde{H^{\frac{1}{2}}}$ is then a LETS with no unsatisfied check node ($b=0$) in a variable-regular LDPC code with girth $g$ and variable degree $d_v$.
We thus have $\mathcal{T}^{0}_{LETS}(d_v;g) \leq n(d_v;\frac{g}{2})$. Now, we show that $\mathcal{T}^{0}_{LETS}(d_v;g) = n(d_v;\frac{g}{2})$.
To the contrary, assume that  $\mathcal{T}^{0}_{LETS}(d_v;g) < n(d_v;\frac{g}{2})$, and let $S$ be a LETS of size $\mathcal{T}^{0}_{LETS}(d_v;g)$ in a variable-regular LDPC Tanner graph $G$ with girth $g$ and variable degree $d_v$,
such that $G(S)$ does not have any unsatisfied check node. The normal graph of $G(S)$ is a graph with girth at least $\frac{g}{2} (\geq 3)$, in which there are $\mathcal{T}^{0}_{LETS}(d_v;g)$ nodes, each with degree $d_v$.
Since the function $n(d_v; k)$ increases monotonically with $k$ \cite{exoo2008dynamic}, this is a contradiction.

$(i)(b)$ Let $G$ be a $(d_v;\frac{g}{2})$-graph. Consider the Tanner graph $\widetilde{G^{\frac{1}{2}}}$. Let $v$ and $u$ be two variable nodes and $c$ be a check node in $\widetilde{G^{\frac{1}{2}}}$ such that $vc,uc\in E(\widetilde{G^{\frac{1}{2}}})$. (For the simplification of notations, we use $vc$ or $cv$ to indicate the undirected edge $\{v,c\}$.) Remove the check node $c$ and its incident edges from  $\widetilde{G^{\frac{1}{2}}}$. Then, add a variable node $w$, $d_v$ check nodes $c_{w,v},c_{w,u},c_{w_1}, \ldots, c_{w_{d_v-2}}$, and the edges $c_{w,v}w, c_{w,v}v, c_{w,u}w, c_{w,u}u, c_{w_1} w, \ldots, c_{w_{d_v-2}} w$, to the graph. The resultant Tanner graph is a LETS with $d_v-2$ unsatisfied check node in a variable-regular LDPC code with variable degree $d_v$ and girth at least $g$.

$(i)(c)$ Let $G$ be a $(d_v;\frac{g}{2})$-graph. Consider the Tanner graph $\widetilde{G^{\frac{1}{2}}}$. Let $v $ be a variable node in $\widetilde{G^{\frac{1}{2}}}$. Remove the variable node $v$ and its incident edges from $\widetilde{G^{\frac{1}{2}}}$. The resultant Tanner graph is a LETS with $b=d_v$ in a variable-regular LDPC code with variable degree $d_v$ and girth at least $g$.

$(i)(d)$ Let $G$ be a $(d_v;\frac{g}{2})$-graph. Consider the Tanner graph $\widetilde{G^{\frac{1}{2}}}$. Let $v_1, \ldots, v_b,\: b = 2i,\: i \geq 1$, be $b$ distinct variable nodes in $\widetilde{G^{\frac{1}{2}}}$ such that $v_k$ and $v_{i+k}$, for any $k=1, \ldots, i$, share check node $c_k$ as a neighbor. (Note that by the condition $b < d_v+2$, one can always find such group of variable nodes.)
Remove the check  nodes $c_1, \ldots, c_i$, and their incident edges from $\widetilde{G^{\frac{1}{2}}}$. Then, add $b=2i$ check nodes $c_{v_1}, \ldots, c_{v_{b}}$ and $b$ edges $c_{v_1}v_1, \ldots, c_{v_b}v_b$ to the Tanner graph. The resultant Tanner graph is a LETS with $b$ degree-one check nodes in a variable-regular LDPC code with variable degree $d_v$ and girth at least $g$.

$(ii)$ Let $G$ be a $(d_v;\frac{g}{2})$-graph. Consider the Tanner graph $\widetilde{G^{\frac{1}{2}}}$. Let $v_1,v_2$ be two variable nodes and $c$ be a check node in $\widetilde{G^{\frac{1}{2}}}$ such that $v_1c,v_2c\in E(\widetilde{G^{\frac{1}{2}}})$. Since $g$ is at least eight, there are $\eta = 2(d_v-1)$ distinct check nodes $c_1, \ldots, c_{\eta}$ in $\widetilde{G^{\frac{1}{2}}}$ such that $v_1c_1, \ldots, v_1c_{\eta/2}, v_2 c_{\eta/2+1}, \ldots,  v_2c_{\eta} \in E(\widetilde{G^{\frac{1}{2}}})$. Also, there are $\eta$ distinct variable nodes $u_1, \ldots, u_{\eta}$ (in addition to $v_1$ and $v_2$) such that $u_1c_1, \ldots, u_{\eta}c_{\eta} \in E(\widetilde{G^{\frac{1}{2}}})$ (otherwise the girth is at most six).
Now, remove the variable  nodes $v_1,v_2$, check node $c$ and their incident edges from $\widetilde{G^{\frac{1}{2}}}$. The resultant Tanner graph is a LETS with $b= 2(d_v-1)$ in a variable-regular LDPC code with variable degree $d_v$ and girth at least $g$. 

$(iii)$ Let $G$ be a $(d_v;\frac{g}{2})$-graph with $d_v \geq 4$ and $g \geq 6$. Consider the Tanner graph $\widetilde{G^{\frac{1}{2}}}$. Let $v,u_1,u_2$ be three distinct variable nodes in $\widetilde{G^{\frac{1}{2}}}$ and $c $ be a check node such that $cu_1,cu_2\in E(\widetilde{G^{\frac{1}{2}}})$. Remove the variable  node $v$ and its incident edges from $\widetilde{G^{\frac{1}{2}}}$. Also, remove the check  node $c$ and its incident edges. Now, add two check nodes $c_{u_1},c_{u_2}$ and the edges $c_{u_1}u_1,c_{u_2}u_2$ to the graph. The resultant Tanner graph is a LETS with $b=d_v+2$ in a variable-regular LDPC code with variable degree $d_v$ and girth at least $g$. (Note that $v$ can share at most one neighboring check node with either $u_1$ or $u_2$. So, after the modifications, both $u_1$ and $u_2$ are still connected to at least two degree-$2$ check nodes.) For the case of $d_v=3$, based on the constraint $g \geq 10$, one can find three distinct variable nodes $v,u_1,u_2$ in $\widetilde{G^{\frac{1}{2}}}$ and a check node $c$ such that $cu_1,cu_2\in E(\widetilde{G^{\frac{1}{2}}})$, and $v$ does not have any common neighbor with $u_1$ or $u_2$. Remove the variable node $v$ and its incident edges, as well as the check  node $c$ and its incident edges from $\widetilde{G^{\frac{1}{2}}}$. Now, add two check nodes $c_{u_1},c_{u_2}$ and the edges $c_{u_1}u_1,c_{u_2}u_2$ to the graph. The resultant Tanner graph is a LETS with $b= d_v+2 = 5$ in a variable-regular LDPC code with $d_v=3$ and girth at least $g$.
}\end{proof}

\subsubsection{Exact value of $\mathcal{T}^{0}_{LETS}(d_v;g)$}

Part $(i)(a)$ of Theorem~\ref{Th92} provides the exact value of $\mathcal{T}^{0}_{LETS}(d_v;g)$ in terms of the size of the corresponding cage. Fig.~\ref{F99}(b) is an example of the construction used in this theorem, where a $(4,0)$ LETS is constructed from the $(3;3)$-cage.

\begin{table}[ht]
\caption{Comparison of the results of Part $(i)(a)$ of Theorem~\ref{Th92} with existing results for $\mathcal{T}^{0}_{LETS}(d_v;g)$}
\begin{center}
\scalebox{1}{
\begin{tabular}{ |l||c|c|c|c|c|c|   }
\hline
$\mathcal{T}^{0}_{LETS}(3;g)$  &  $g=$6  & $g=$8    & $g=$10    &$g=$12    & $g=$14     &$g=$16       \\
\hline
\hline
Lower bound  of Theorem~\ref{lowf}   & 4   & 6   & 10    & 14   & 22    & 30        \\
\hline
Known exact  value~\cite{hashemilower},~\cite{hashemi2018}             & 4       &6         & 10       &--          &--         &--     \\
\hline
 Theorem~\ref{Th92}   &  4   & 6   &  10    & 14   & 24    & 30        \\
\hline
 \multicolumn{7}{|l|}{$\mathcal{T}^{0}_{LETS}(4;g)$} \\
\hline
Lower bound  of Theorem~\ref{lowf}    &  5   &   8   &   17    &   26   &  53    &   80        \\
\hline
Known exact  value~\cite{hashemilower},~\cite{hashemi2018}           & 5       &8          & 19       &--          &--         &--     \\
\hline
Theorem~\ref{Th92}   &   5   &   8   &    19    &   26   &   67    &   80        \\
\hline
 \multicolumn{7}{|l|}{$\mathcal{T}^{0}_{LETS}(5;g)$} \\
\hline
Lower bound  of Theorem~\ref{lowf}    &  6   &   10   &   26    &   42   &  106    &   170        \\
\hline
Known exact  value~\cite{hashemilower},~\cite{hashemi2018}          & 6       &10         & 30       &--          &--         &--     \\
\hline
Theorem~\ref{Th92}    &   6   &   10   &    30    &   42   &   152    &   170        \\
\hline
 \multicolumn{7}{|l|}{$\mathcal{T}^{0}_{LETS}(6;g)$} \\
\hline
Lower bound  of Theorem~\ref{lowf}      &  7    &   12    &   37    &   62   &  187    &   312        \\
\hline
Known exact  value~\cite{hashemilower},~\cite{hashemi2018}            & 7       &12         & --       &--          &--         &--     \\
\hline
Theorem~\ref{Th92}   &   7   &   12   &    40    &   62   &  294    &   312        \\
\hline
\end{tabular}
}
\end{center}
\label{Table2}
\end{table}

In Table \ref{Table2}, we compare the results of Part $(i)(a)$ of Theorem~\ref{Th92} with previously existing results for different values of $d_v$ and $g$. As expected, the results of Theorem~\ref{Th92} match any previous result obtained in
\cite{hashemilower} and~\cite{hashemi2018} based on $dpl$ characterization on the exact value of  $\mathcal{T}^{0}_{LETS}(d_v;g)$. Due to the high complexity, however, the results for larger girth values are not
reported in \cite{hashemilower} and~\cite{hashemi2018}.\footnote{The time limitation for calculations, indicated in~\cite{hashemilower}, is a day on a desktop computer with $2.4$-GHz CPU and $8$-GB RAM.}
These correspond to entries with symbol ``-''. We note that for some of the cases where the exact value of $\mathcal{T}^{0}_{LETS}(d_v;g)$ was not previously known, our results have a rather large gap with the lower bound of
Theorem~\ref{lowf}. For some other cases, the lower bound of  Theorem~\ref{lowf} is tight. In fact, based on the results of Table~\ref{Table2}, one may suggest that the lower bound of Theorem~\ref{lowf} is always tight for
even values of $g/2$. This, however, is not the case. For example, based on Theorem~\ref{Th92}, the exact value of $\mathcal{T}^{0}_{LETS}(3;20)$ is $72$, while the lower bound of Theorem~\ref{lowf}, for this case, is $62$.

\begin{rem}\label{R1}
The number of non-isomorphic $(x;y)$-cages for small values of $x$ and $y$ are known, see, e.g., \cite{exoo2008dynamic}. Thus, by Part $(i)(a)$ of Theorem \ref{Th92}, we can determine
the number of non-isomorphic $(n(d_v;g/2),0)$ LETSs  in Tanner graphs with girth $g$ and variable degree $d_v$. Note that if $a< n(d_v;g/2)$, then there is no $(a,0)$ LETSs in the Tanner graphs with girth $g$ and variable degree $d_v$, and
thus $n(d_v;g/2)$ is the smallest size of LETSs with $b=0$ (lowest weight of elementary codewords). The girth of minimum size  LETSs with $b=0$ must be $g$,\footnote{Because otherwise, the normal graph of such a LETS will be a graph with $n(d_v;g/2)$ nodes, where all the nodes have degree $d_v$ and the girth of the graph is larger than $g/2$. This contradicts the result of Erd\H{o}s and Sachs (see \cite{exoo2008dynamic}, page 5) which states that if $G$ has the minimum number of nodes for a $k$-regular graph with girth at least $g$, then the girth of $G$ is exactly $g$.} and thus there is a one-to-one correspondence between non-isomorphic
$(d_v,g/2)$-cages and non-isomorphic $(n(d_v;g/2),0)$ LETSs in Tanner graphs with girth $g (\geq 6)$ and variable degree $d_v (\geq 3)$.
For instance, there is only one $(3;5)$-cage and that cage, which has $10$ nodes, is the so-called ``Petersen graph'' \cite{exoo2008dynamic}. The unique $(10,0)$
LETS corresponding to the Petersen graph is shown in Fig. \ref{F88}. (In trapping set subgraphs, variable nodes, satisfied check nodes and unsatisfied check nodes are represented by circles,
full squares and empty squares, respectively.)
\end{rem}

\begin{figure}[]
\centering
\includegraphics [width=0.25\textwidth]{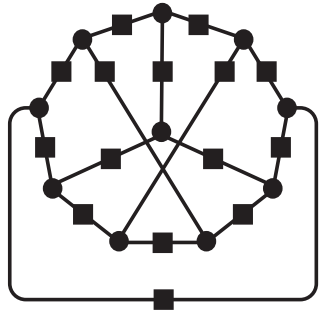}
\caption{The unique $(10,0)$ LETS in Tanner graphs with girth $10$ and $d_v=3$. }
\label{F88}
\end{figure}

\subsubsection{Upper bounds on $\mathcal{T}^{1}_{LETS}(d_v;g)$}

We know, based on Lemma~\ref{lem:cannot}, that there are no trapping sets with an odd value of $b$ in a Tanner graph with an even value of $d_v$. We use Part $(i)(b)$ of Theorem~\ref{Th92} to find an upper bound on $\mathcal{T}^{1}_{LETS}(3;g)$. Fig.~\ref{F99}(c) is an example of the construction used in Part $(i)(b)$ of Theorem~\ref{Th92}, where a $(5,1)$ LETS is constructed from the $(3;3)$-cage.

For $d_v=5$, we use the following result, whose proof is provided in the appendix, to derive an upper bound on $\mathcal{T}^{1}_{LETS}(5;g)$.

\begin{fact}
\label{F1}
For any $g \geq 6$, we have $\mathcal{T}^{1}_{LETS}(5;g) \leq n'(5,6;\frac{g}{2})$.
\end{fact}

For $\frac{g}{2}=3$, it was shown in \cite{MR642638} that $n'(5,6;3)=7$. To the best of our knowledge, there is no result available in the literature for $\frac{g}{2}=4$.
For $\frac{g}{2}=5$, it was shown in \cite{cage} that $n'(5,6;5)=31$. For the cases where $\frac{g}{2} \geq 6$, also, to the best of our knowledge, not much
progress in computing $n'(5,6;\frac{g}{2})$ (or $n(5,6;g/2)$) has been made.

In Table \ref{TableNN2}, we compare the upper bounds of Part $(i)(b)$ of Theorem~\ref{Th92} and Fact~\ref{F1} with existing exact values and lower bounds. As can be seen in the table, for $d_v=3$, and for all the cases where the exact value of $\mathcal{T}^{1}_{LETS}(3;g)$ is known, our upper bounds match the exact value. For $d_v=5$, based on the existing graph theoretical results, we are only able to provide our upper bounds for $g=6$ and $g=10$.
For $g=6$, which is the only case out of the two where the exact value of $\mathcal{T}^{1}_{LETS}(3;g)$ is known, again our bound is tight.

\begin{table}[ht]
\caption{Comparison of upper bounds of Part $(i)(b)$ of Theorem~\ref{Th92} and fact~\ref{F1} with existing results for $\mathcal{T}^{1}_{LETS}(d_v;g)$}
\begin{center}
\scalebox{1}{
\begin{tabular}{ |l||c|c|c|c|c|c|   }
\hline
$\mathcal{T}^{1}_{LETS}(3;g)$  &  $g=$6  & $g=$8    & $g=$10    &$g=$12    & $g=$14     &$g=$16       \\
\hline
\hline
Lower bound  of Theorem~\ref{lowf}   &  5   &   5   &   9    &   11   &  19    &   23        \\
\hline
Known exact  value~\cite{hashemilower}            & 5       &7          & 11      &--          &--         &--    \\
\hline
Upper bound of Theorem~\ref{Th92}   &   5   &   7   &    11    &   15   &   25    &   31        \\
\hline
 \multicolumn{7}{|l|}{$\mathcal{T}^{1}_{LETS}(5;g)$} \\
\hline
Lower bound  of Theorem~\ref{lowf}   &  7   &   9   &   25    &   37   &  101    &   149        \\
\hline
Known exact  value~\cite{hashemilower}           & 7       &13          & --      &--          &--         &--    \\
\hline
Upper bound of Fact~\ref{F1}   &   7   &   --   &    31    &   --   &   --    &   --        \\
\hline
\end{tabular}
}
\end{center}
\label{TableNN2}
\end{table}

\begin{ex}\label{R2}
As we discussed in Remark~\ref{R1}, there is a unique $(10,0)$ LETS structure in the Tanner graphs with girth $10$ and $d_v = 3$ (see Fig. \ref{F88}).
The normal graph of the unique $(10,0)$ LETS structure is the Petersen graph. The Petersen graph is edge transitive.\footnote{An edge-transitive graph is a graph $G$ such that, given any two edges $e_1$ and $e_2$ of $G$, there is an automorphism of $G$ that maps $e_1$ to $e_2$.} Thus, by applying the procedure which was described in Part $(i)(b)$ of Theorem \ref{Th92} to Peterson graph, we obtain a unique $(11,1)$ LETS structure of Tanner graphs with girth $10$ and $d_v = 3$ (see Fig. \ref{F2}).
It is known through computer search that the structure in Fig.~\ref{F2} is the only $(11,1)$ LETS structure in Tanner graphs with girth $10$ and $d_v = 3$.
\end{ex}

\begin{figure}[]
\centering
\includegraphics [width=0.25\textwidth]{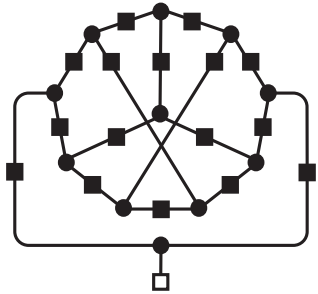}
\caption{The unique (11,1) LETS in the Tanner graphs with girth 10 and $d_v=3$. }
\label{F2}
\end{figure}

\subsubsection{Upper bounds on $\mathcal{T}^{2}_{LETS}(d_v;g)$}\label{subsectionB}

Based on Part $(i)(d)$ of Theorem~\ref{Th92}, we derive an upper bound on $\mathcal{T}^{2}_{LETS}(d_v;g)$, for any $g \geq 6$ and any $d_v \geq 3$. In Table \ref{Table5}, we compare this upper bound with the existing lower bound and exact values for $\mathcal{T}^{2}_{LETS}(d_v;g)$. An inspection of Table~\ref{Table5} reveals that for every case where the exact value of $\mathcal{T}^{2}_{LETS}(d_v;g)$ is known, our upper bound matches the exact value and is thus tight.

\begin{table}[ht]
\caption{Comparison of the upper bound of Part $(i)(d)$ of Theorem~\ref{Th92} with the existing results on $\mathcal{T}^{2}_{LETS}(d_v;g)$}
\begin{center}
\scalebox{1}{
\begin{tabular}{ |l||c|c|c|c|c|c|   }
\hline
$\mathcal{T}^{2}_{LETS}(3;g)$  &  $g=$6  & $g=$8    & $g=$10    &$g=$12    & $g=$14     &$g=$16       \\
\hline
\hline
Lower bound  of Theorem~\ref{lowf}   &  4   &   4   &   6    &   8   &  12    &   16        \\
\hline
Known exact  value~\cite{hashemilower}             & 4       &6          & 10       &--          &--         &--     \\
\hline
Upper bound of Theorem~\ref{Th92}   &   4   &   6   &    10    &   14   &   24    &   30        \\
\hline
 \multicolumn{7}{|l|}{$\mathcal{T}^{2}_{LETS}(4;g)$} \\
\hline
Lower bound  of Theorem~\ref{lowf}   &  4   &   6   &   12    &   18   &  36    &   54        \\
\hline
Known exact  value~\cite{hashemilower}              & 5       &8          & 19       &--          &--         &--    \\
\hline
Upper bound of Theorem~\ref{Th92}   &   5   &   8   &    19    &   26   &   67    &   80        \\
\hline
 \multicolumn{7}{|l|}{$\mathcal{T}^{2}_{LETS}(5;g)$} \\
\hline
Lower bound  of Theorem~\ref{lowf}   &  5   &   8   &   20    &   32   &  80    &   128        \\
\hline
Known exact  value~\cite{hashemilower}             & 6       &10         & --      &--          &--         &--     \\
\hline
Upper bound of Theorem~\ref{Th92}   &   6   &   10   &    30    &   42   &   152    &   170        \\
\hline
 \multicolumn{7}{|l|}{$\mathcal{T}^{2}_{LETS}(6;g)$} \\
\hline
Lower bound  of Theorem~\ref{lowf}   &  7   &   10   &   30    &   50   &  150    &   250        \\
\hline
Known exact  value~\cite{hashemilower}          & 7       &12          & --       &--          &--         &--     \\
\hline
Upper bound of Theorem~\ref{Th92}   &   7   &   12   &    40    &   62   &   294    &   312        \\
\hline
\end{tabular}
}
\end{center}
\label{Table5}
\end{table}

\subsubsection{Upper bounds on $\mathcal{T}^{3}_{LETS}(d_v;g)$}
\label{subsectionC}

We know, based on Lemma~\ref{lem:cannot}, that there are no trapping sets with $b = 3$ in a Tanner graph with an even value of $d_v$. For $d_v = 3$ and $d_v=5$, we use the results derived in Parts $(i)(c)$ and $(i)(b)$ of Theorem~\ref{Th92} to upper bound $\mathcal{T}^{3}_{LETS}(3;g)$ and $\mathcal{T}^{3}_{LETS}(5;g)$, respectively, for any $g \geq 6$.

In Table \ref{Table6}, we compare the upper bounds of Theorem~\ref{Th92} with the existing lower bound and exact results for $\mathcal{T}^{3}_{LETS}(d_v;g)$.
Again, as can be seen in the table, for all the known values of $\mathcal{T}^{3}_{LETS}(d_v;g)$, the derived upper bounds are tight with a large gap to the existing lower bound, particularly for larger girth values.

\begin{table}[ht]
\caption{Comparison of the upper bounds of Theorem~\ref{Th92} with the existing results for $\mathcal{T}^{3}_{LETS}(d_v;g)$}
\begin{center}
\scalebox{1}{
\begin{tabular}{ |l||c|c|c|c|c|c|   }
\hline
$\mathcal{T}^{3}_{LETS}(3;g)$  &  $g=$6  & $g=$8    & $g=$10    &$g=$12    & $g=$14     &$g=$16       \\
\hline
\hline
Lower bound  of Theorem~\ref{lowf}   &  3   &   3   &   5    &   5   &  9    &   9        \\
\hline
Known exact  value~\cite{hashemilower}          & 3      &5          & 9      &--          &--         &--    \\
\hline
Upper bound of Theorem~\ref{Th92}   &   3   &   5   &   9     &   13   &   23    &   29        \\
\hline
 \multicolumn{7}{|l|}{$\mathcal{T}^{3}_{LETS}(5;g)$} \\
\hline
Lower bound  of Theorem~\ref{lowf}   &  5   &   7   &   19    &   27   &  75    &   107        \\
\hline
Known exact  value~\cite{hashemilower}      & 7      &11          & --      &--          &--         &--    \\
\hline
Upper bound of Theorem~\ref{Th92}    &   7   &   11   &   31     &   42   &   153    &   171        \\
\hline
\end{tabular}
}
\end{center}
\label{Table6}
\end{table}

\subsubsection{Upper bounds on $\mathcal{T}^{4}_{LETS}(d_v;g)$}
\label{subsectionA}

First, we derive the following simple lower bound on $\mathcal{T}^{4}_{LETS}(3;g)$.

\begin{lem}
For any $g$, we have $\mathcal{T}^{4}_{LETS}(3;g)\geq 4$.
\label{lem53}
\end{lem}

\begin{proof}{
This follows from the fact that unsatisfied check nodes in a LETS have degree one and that each variable node of degree three can be connected to at most one unsatisfied check node.
}\end{proof}

For $g\leq 8$, the lower bound of Lemma~\ref{lem53} is tight. To see this, let $G$ be a cycle of length four. Consider the Tanner graph $\widetilde{G^{\frac{1}{2}}}$. Then, for each variable node $v$ in $\widetilde{G^{\frac{1}{2}}}$, add a  check node $c_v$ and connect $v$ to $c_v$.  The resultant graph is a LETS in a variable-regular LDPC code with $d_v=3$ and with $g \leq 8$.

To complement the above result, we use Part $(ii)$ of Theorem~\ref{Th92} and derive an upper bound on $\mathcal{T}^{4}_{LETS}(3;g)$ for any $g \geq 8$. We also use Parts $(i)(c)$ and $(i)(d)$ of Theorem~\ref{Th92} to derive upper bounds on $\mathcal{T}^{4}_{LETS}(4;g)$ and $\mathcal{T}^{4}_{LETS}(d_v;g)$, for $d_v \geq 5$, respectively.

In Table \ref{Table8}, we compare the upper bounds of Theorem~\ref{Th92} with the existing results for $\mathcal{T}^{4}_{LETS}(3;g)$ and $\mathcal{T}^{4}_{LETS}(4;g)$. Note that for the case of $d_v=3$, there is no lower bound presented in \cite{hashemilower}. The comparison, yet again, shows that the derived upper bounds are tight for all the cases where the exact value of $\mathcal{T}^{4}_{LETS}(d_v;g)$ is known and that there is a large gap with the existing lower bound of Theorem~\ref{lowf}.

\begin{table}[ht]
\caption{Comparison of the bounds of Lemma~\ref{lem53} and Theorem~\ref{Th92} with the existing results for $\mathcal{T}^{4}_{LETS}(d_v;g)$}
\begin{center}
\scalebox{1}{
\begin{tabular}{ |l||c|c|c|c|c|c|   }
\hline
$\mathcal{T}^{4}_{LETS}(3;g)$  &  $g=$6  & $g=$8    & $g=$10    &$g=$12    & $g=$14     &$g=$16       \\
\hline
\hline
Lower bound of Lemma~\ref{lem53}  & 4   & 4   & 4    & 4   & 4    & 4        \\
\hline
Known exact  value~\cite{hashemilower}           & 4      &4          & 8      &--          &--         &--    \\
\hline
Upper bound of Theorem~\ref{Th92}   & --   &   4   &   8     &   12   &   22    &  28        \\
\hline
 \multicolumn{7}{|l|}{$\mathcal{T}^{4}_{LETS}(4;g)$} \\
\hline
Lower bound  of Theorem~\ref{lowf}   &  3   &   4   &   7    &   10   &  19    &   28        \\
\hline
Known exact  value~\cite{hashemilower}             & 4      &7          & 18      &--          &--         &--    \\
\hline
Upper bound of Theorem~\ref{Th92}   &   4   &   7   &   18     &   25   &   66    &  79        \\
\hline
\end{tabular}
}
\end{center}
\label{Table8}
\end{table}

\subsubsection{Upper bounds on $\mathcal{T}^{5}_{LETS}(d_v;g)$}

First, we present the following lower bound on $\mathcal{T}^{5}_{LETS}(3;g)$, whose proof is similar to that of Lemma~\ref{lem53}.

\begin{lem}
For any $g$, we have $\mathcal{T}^{5}_{LETS}(3;g)\geq 5$.
\label{lem23}
\end{lem}

For the case of $g\leq 10$, the lower bound of Lemma~\ref{lem23} is tight. Let $G$ be a cycle of length five. Consider the Tanner graph $\widetilde{G^{\frac{1}{2}}}$, and for each variable node $v$ in $\widetilde{G^{\frac{1}{2}}}$, add a  check node $c_v$ and connect $v$ to $c_v$.  The resultant graph is a LETS with $b=5$ in variable-regular graphs with $d_v=3$ and $g \leq 10$. 

To complement the above result, we use the result of Part $(iii)$ of Theorem~\ref{Th92}, and derive an upper bound on $\mathcal{T}^{5}_{LETS}(3;g)$ for any $g \geq 12$. We also use Part $(i)(c)$ of Theorem~\ref{Th92} to derive an upper bound on $\mathcal{T}^{5}_{LETS}(5;g)$, for any $g \geq 6$.

Table \ref{TableN1} shows the comparison between the results of Lemma~\ref{lem23} and Theorem~\ref{Th92}, on the one hand, and the existing exact values and lower bounds on $\mathcal{T}^{5}_{LETS}(d_v;g)$, on the other hand. Similar conclusions as in those related to previous tables can be drawn here.

\begin{table}[ht]
\caption{Comparison of the bounds of  Lemma~\ref{lem23} and Theorem~\ref{Th92} with the existing results for $\mathcal{T}^{5}_{LETS}(d_v;g)$}
\begin{center}
\scalebox{1}{
\begin{tabular}{ |l||c|c|c|c|c|c|   }
\hline
$\mathcal{T}^{5}_{LETS}(3;g)$  &  $g=$6  & $g=$8    & $g=$10    &$g=$12    & $g=$14     &$g=$16       \\
\hline
\hline
Lower bound of Lemma~\ref{lem23}   &  5   &   5   &   5    &   5   &  5    &   5        \\
\hline
Known exact  value~\cite{hashemilower}             & 5      &5          & 5      &--          &--         &--    \\
\hline
Upper bound of  Theorem~\ref{Th92}   &   --   &   --   &    --     &   13   &  23    &  29        \\
\hline
 \multicolumn{7}{|l|}{$\mathcal{T}^{5}_{LETS}(5;g)$} \\
\hline
Lower bound of Theorem~\ref{lowf}  &  5   &   5   &   13    &   17   &  33    &   65        \\
\hline
Known exact  value~\cite{hashemilower}             & 5      &9          & --      &--          &--         &--    \\
\hline
Upper bound of  Theorem~\ref{Th92}   &   5   &   9   &   29     &   41   &   151    &  169        \\
\hline
\end{tabular}
}
\end{center}
\label{TableN1}
\end{table}

\begin{rem}
Many of the upper bounds derived in this subsection match the exact values obtained in~\cite{hashemilower},~\cite{hashemi2018}. Based on the results of~\cite{hashemi2015new}, for a large number of such cases, there is only a single structure in the corresponding class of LETSs. For such cases, our results clearly identify those unique structures. Examples of such cases include the unique structures of $(4,0)$, $(4,2)$ and $(5,1)$ LETSs in variable-regular LDPC codes with $d_v=3$ and $g=6$. The same applies to the unique LETS structures in $(6,0)$, $(7,1)$ and $(6,2)$ classes of codes with $d_v=3$ and $g=8$, as well as the structures $(5,0)$ and $(5,2)$ for codes with $d_v=4$ and $g=6$. Note that in all the above examples, the class of interest is the one with minimum size $a$ for the given value of $b$. These classes are thus among the most dominant in the error floor region.
\end{rem}

\subsection{Elementary trapping sets with leafs (ETSLs)}
\label{sec4}

Elementary trapping sets with leafs can be partitioned into two categories: (a) Those with no cycle, and (b) those that contain at least one cycle. The two categories are labeled as ETSL$_2$ and ETSL$_1$, respectively, in \cite{hashemilower}.
Trapping sets in  ETSL$_2$ have a tree structure, and only exist for $b \geq d_v$. In such cases, the size of ETSL$_2$ structures is fixed and is equal to $(b-2)/(d_v-2)$~\cite{hashemilower}. In the rest of this subsection, therefore, our focus will be on  ETSL$_1$s.
For simplicity, however, we refer to ETSL$_1$ structures as ETSL.

We note that in variable-regular Tanner graphs there is no ETSL with $b=0$ or $b=1$. Also, for $b\geq 2$, ETSLs exist only if $d_v\leq b+1$ \cite{hashemilower}.
In the following theorem, we find general upper and lower bounds on the minimum size of ETSLs in terms of  the minimum size of LETSs.

\begin{theo}\label{Th25}
For any $g \geq 6$ and any $d_v \geq 3$, we have (i) $\mathcal{T}^{b+d_v-2}_{ETSL}(d_v;g) \leq \mathcal{T}^{b}_{LETS}(d_v;g)+1$, for any $b \geq 1$, and (ii) $\mathcal{T}^{d_v-1}_{ETSL}(d_v;g) \geq \mathcal{T}^{1}_{LETS}(d_v;g)+1$.
\end{theo}

\begin{proof}{
$(i)$ Let $S$ be a LETS with $b\geq 1$ in a variable-regular LDPC code. The induced subgraph of $S$, $G(S)$, has at least one degree-1 check node. Call it $c$. Add a variable node $v$, $d_v-1$ check nodes $c_1,\ldots, c_{d_v-1}$ and the edges $cv, c_1v,\ldots, c_{d_v-1}v$ to the trapping set. The resultant graph is an ETSL with $b+d_v-2$ unsatisfied check nodes.

$(ii)$ Let $S$ be an ETSL with $b=d_v-1$ in a variable-regular LDPC code with variable degree $d_v$ and grith $g$. Since the normal graph of $G(S)$ has a leaf, there is a variable node such as $v$ in $S$ such that the degree of only one of its neighboring check nodes is two. (Since $b=d_v-1$, this variable node is unique.)
Call that check node $c_1$. Assume that $\{c_1,\ldots,c_{d_v}\}$ is the set of check nodes that are adjacent to the variable node $v$. Remove the nodes $v, c_2, \ldots, c_{d_v}$ and their incident edges form $G(S)$. The resultant graph
is a LETS with $b=1$ and girth at least $g$. 
}\end{proof}

Combining the upper and the lower bounds of Theorem~\ref{Th25} for $b=1$, we obtain the following result. (Note that a variable-regular graph with even $d_v$ cannot have a TS with odd value of $b$.)

\begin{cor}
For any $g \geq 6$ and any odd $d_v \geq 3$, we have $\mathcal{T}^{d_v-1}_{ETSL}(d_v;g) = \mathcal{T}^{1}_{LETS}(d_v;g)+1$.
\label{cor43}
\end{cor}

\begin{rem}{
We note that the general lower bound of (\ref{ineq1}) for ETSs is also applicable to ETSLs. Using the result of Corollary~\ref{cor43}, however, we can improve the lower bound for ETSLs, i.e., we can use (\ref{ineq1}) to obtain a lower bound on $\mathcal{T}^{1}_{LETS}(d_v;g)$, and then use Corollary~\ref{cor43} to derive a lower bound on $\mathcal{T}^{d_v-1}_{ETSL}(d_v;g)$. As an example, the lower bound of (\ref{ineq1}) for $\mathcal{T}^{2}_{ETSL}(3;10)$ is $6$. On the other hand, the lower bound of (\ref{ineq1}) for $\mathcal{T}^{1}_{LETS}(3;10)$ is $9$. Using Corollary~\ref{cor43}, this translates to the improved lower bound of $10$ for $\mathcal{T}^{2}_{ETSL}(3;10)$.
These improved lower bounds for $d_v=3, 5$, and different girth values are presented in Table \ref{Table15}.
}
\label{rem12}
\end{rem}

\begin{ex}
In Example~\ref{R2}, we demonstrated that there is a unique $(11,1)$ LETS structure in Tanner graphs with girth 10 and $d_v = 3$. By the procedure presented in Theorem \ref{Th25}, we find a unique $(12,2)$ ETSL structure in Tanner graphs with girth 10 and $d_v = 3$. This ETSL is shown in Fig. \ref{F3}.
\end{ex}

\begin{figure}[]
\centering
\includegraphics [width=0.25\textwidth]{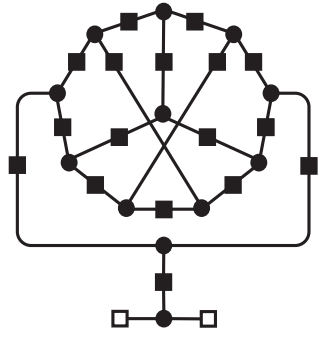}
\caption{The unique $(12,2)$ ETSL structure in Tanner graphs with girth $10$ and $d_v=3$.}
\label{F3}
\end{figure}

In the following corollary, using Part $(i)$ of Theorem~\ref{Th25} and the results of Subsection~\ref{sec3}, we provide upper bounds on the smallest size of ETSLs in terms of cage sizes, for different $b$ values.

\begin{cor}
For any $g \geq 6$, we have (i) $\mathcal{T}^{2}_{ETSL}(3;g)\leq n(3;\frac{g}{2})+2$, (ii) $\mathcal{T}^{3}_{ETSL}(3;g)\leq n(3;\frac{g}{2})+1$, (iii) $\mathcal{T}^{4}_{ETSL}(3;g)\leq n(3;\frac{g}{2})$, $\mathcal{T}^{4}_{ETSL}(4;g)\leq n(4;\frac{g}{2})+1$, and $\mathcal{T}^{4}_{ETSL}(5;g) \leq n'(5,6;\frac{g}{2}) + 1$. Moreover, (iv)(a) for any $g \geq 8$, we have $\mathcal{T}^{5}_{ETSL}(3;g)\leq n(3;\frac{g}{2})-1$, and for $g=6$, we have $\mathcal{T}^{5}_{ETSL}(3;6)\leq 5$; and (iv)(b) for any $g\geq 6$, we have $\mathcal{T}^{5}_{ETSL}(5;g)\leq n(5;\frac{g}{2})+1$.
\label{cor19}
\end{cor}

\begin{table}[ht]
\caption{Comparison of the results of this work with existing results for $\mathcal{T}^{b}_{ETSL}(d_v;g)$}
\begin{center}
\scalebox{1}{
\begin{tabular}{ |l||c|c|c|c|c|c|   }
\hline
$\mathcal{T}^{2}_{ETSL}(3;g)$  &  $g=$6  & $g=$8    & $g=$10    &$g=$12    & $g=$14     &$g=$16       \\
\hline
Lower bound of Theorem~\ref{lowf}  &  4   &   4   &  6    &   8   &  12    &  16        \\
\hline
New Lower bound (Remark~\ref{rem12})  &   6   &   6   &   10     &   12   &   20    &  24        \\
\hline
Known exact  value~\cite{hashemilower}  &6 & 8 & 12 & -- & -- & --\\
\hline
Upper bound of Corollary~\ref{cor19}   &   6   &   8   &    12    &   16   &   26    &   32        \\
\hline
\multicolumn{7}{|l|}{$\mathcal{T}^{3}_{ETSL}(3;g)$} \\
\hline
Lower bound of Theorem~\ref{lowf}   &  3   &   3   &  5    &   5   &  9    &  9        \\
\hline
Known exact  value~\cite{hashemilower}             & 5      &7          & 11      &--          &--         &--    \\
\hline
Upper bound of Corollary~\ref{cor19}    &   5   &   7   &   11     &   15   &   25    &  31        \\
\hline
\multicolumn{7}{|l|}{$\mathcal{T}^{4}_{ETSL}(3;g)$} \\
\hline
Known exact  value~\cite{hashemilower}             & 4     &6          & 10      &--          &--         &--    \\
\hline
Upper bound of Corollary~\ref{cor19}   &   4   &   6   &   10     &   14   &   24    &  30        \\
\hline
\multicolumn{7}{|l|}{$\mathcal{T}^{4}_{ETSL}(4;g)$ } \\
\hline
Lower bound of Theorem~\ref{lowf}   &  3   &   4   &  7    &   10   &  19   &  28        \\
\hline
Known exact  value~\cite{hashemilower}             & 6    & 9         & 20      &--          &--         &--    \\
\hline
Upper bound of Corollary~\ref{cor19}   &   6   &   9   &   20     &   27   &   68    &  81        \\
\hline
\multicolumn{7}{|l|}{$\mathcal{T}^{4}_{ETSL}(5;g)$ } \\
\hline
Lower bound of Theorem~\ref{lowf}    &  4   &  6   &  14    &   22   &  54    &  86        \\
\hline
New Lower bound (Remark~\ref{rem12})  &   8   &   10   &   26     &   38   &   102    &  150        \\
\hline
Known exact  value ~\cite{hashemilower}  &8 & 14 & -- & -- & -- & --\\
\hline
Upper bound of Corollary~\ref{cor19}   &   8   &   --   &    32    &   --   &   --    &   --        \\
\hline
\multicolumn{7}{|l|}{$\mathcal{T}^{5}_{ETSL}(3;g)$} \\
\hline
Known exact  value~\cite{hashemilower}       & 5    & 5         & 9      &--          &--         &--    \\
\hline
Upper bound of Corollary~\ref{cor19}   &   5   &   5   &   9     &   13   &   23    &  29        \\
\hline
\multicolumn{7}{|l|}{$\mathcal{T}^{5}_{ETSL}(5;g)$} \\
\hline
Lower bound of Theorem~\ref{lowf}   &  5   &   5   &  13    &   17   &   33    &  65        \\
\hline
Known exact  value~\cite{hashemilower}       & 7    & 11         & --      &--          &--         &--    \\
\hline
Upper bound of Corollary~\ref{cor19}   &   7   &   11   &   31     &   43   &   153    &   171        \\
\hline
\end{tabular}
}
\end{center}
\label{Table15}
\end{table}

We have presented the upper bounds derived in Corollary~\ref{cor19} on $\mathcal{T}^{b}_{ETSL}(d_v;g)$ for different values of $b$, $d_v$ and $g$ in Table~\ref{Table15}.
As can be seen, these bounds in every case where the exact value of $\mathcal{T}^{b}_{ETSL}(d_v;g)$ is known, are equal to the exact value.
Table~\ref{Table15} also shows that there is a large gap between the derived upper bounds and the existing lower bound of Theorem~\ref{lowf} (wherever this bound is available). This is particularly the case for larger $g$ values. This implies that the lower bound of Theorem~\ref{lowf} is rather loose.

\subsection{Non-elementary trapping sets (NETSs)}
\label{sec5}

In the following, we first derive general upper bounds on $\mathcal{T}^{b}_{NETS}(d_v;g)$ for different values of $g \geq 6$, $d_v \geq 3$ and $b \geq 0$. We then improve some of these bounds for specific values of $g$, $d_v$ or $b$.

\begin{theo}\label{Th90}
(i) For any $g \geq 6$ and any $d_v \geq 3$,  we have (a) $\mathcal{T}^{0}_{NETS}(d_v;g)\leq 2n(d_v;\frac{g}{2})$, and (b) $\mathcal{T}^{0}_{NETS}(d_v;g)\leq n(d_v;\frac{g+4}{2})-2$.

(ii) For any $g \geq 6$, we have $\mathcal{T}^{b}_{NETS}(d_v;g)\leq n(d_v;\frac{g+2}{2})-1$, for any even value of $d_v \geq 4$ and any even value of $0 \leq b \leq d_v-2$, or for any odd value of $d_v \geq 3$ and any odd value of $1 \leq b \leq d_v - 2$.

(iii) For any $g \geq 6$,  we have (a) $\mathcal{T}^{2(d_v-1)}_{NETS}(d_v;g)\leq n(d_v;\frac{g }{2})+2$, for any $d_v \geq 3$, (b) $\mathcal{T}^{d_v}_{NETS}(d_v;g)\leq n(d_v;\frac{g }{2})+1$, for any $d_v \geq 3$, (c) $\mathcal{T}^{d_v+2}_{NETS}(d_v;g)\leq n(d_v;\frac{g }{2})+1$, for any $d_v \geq 3$, and (d) $\mathcal{T}^{d_v-1}_{NETS}(d_v;g)\leq n(d_v;\frac{g+2}{2})$, for any odd value $d_v \geq 3$.
\end{theo}

\begin{proof}{
$(i)(a)$ Let $G$ be a $(d_v;\frac{g }{2})$-graph. Consider the Tanner graph $\widetilde{G^{\frac{1}{2}}}$. Let $v_1,v_2$ be two variable nodes in $\widetilde{G^{\frac{1}{2}}}$, and  $c $ be a check node such that $cv_1,cv_2 \in E(\widetilde{G^{\frac{1}{2}}})$. Consider the union of two copies of the graph  $\widetilde{G^{\frac{1}{2}}}$. Remove node $c$ and its incident edges from both copies of  $\widetilde{G^{\frac{1}{2}}}$. Add a check node $c'$ to the graph and connect it the variable nodes  $v_1,v_2$ in both copies of $\widetilde{G^{\frac{1}{2}}}$. The resultant Tanner graph is a NETS with $b=0$ in a variable-regular LDPC code with variable degree $d_v$ and girth at least $g$.

$(i)(b)$ Let $G$ be a $(d_v;\frac{g+4}{2})$-graph. Consider the Tanner graph $\widetilde{G^{\frac{1}{2}}}$. Let $v_1,v_2$ be two variable nodes in $\widetilde{G^{\frac{1}{2}}}$, and  $c$ be  a check node such that $cv_1,cv_2 \in E(\widetilde{G^{\frac{1}{2}}})$. Since the girth of $\widetilde{G^{\frac{1}{2}}}$ is at least ten, there are $2(d_v-1)$ distinct variable nodes $v'_1,\ldots,v'_{d_v-1},v''_1,\ldots,v''_{d_v-1}$ in $\widetilde{G^{\frac{1}{2}}}$ such that $v_1$ has a common neighbor $c'_j$ with $v'_j$, for every $1 \leq j \leq d_v-1$, and $v_2$ has a common neighbor $c''_j$ with $v''_j$,  for every $1 \leq j \leq d_v-1$. Remove the nodes $v_1,v_2, c, c'_1, \ldots, c'_{d_v-1}, c''_1, \ldots, c''_{d_v-1}$, and their incident edges from $\widetilde{G^{\frac{1}{2}}}$. Next, add a check node $c_1$ and the edges $v'_1c_1,\ldots, v'_{d_v-1}c_1,v''_1c_1, \ldots, v''_{d_v-1}c_1$ to the graph. The resultant Tanner graph is a NETS with $b=0$ in a variable-regular LDPC code with variable degree $d_v$ and girth at least $g$.

$(ii)$ Let $G$ be a $(d_v;\frac{g+2}{2})$-graph. Consider the Tanner graph $\widetilde{G^{\frac{1}{2}}}$. Let $v,v_1,\ldots,v_{d_v}$ be $d_v+1$ variable nodes  in $\widetilde{G^{\frac{1}{2}}}$,  such that $v$ has a common neighbor $c_i$ with $v_i$, $i=1,\ldots,d_v$. Let $b > 0$. Remove the nodes $v,c_1, \ldots, c_{d_v-b+1}$, and their incident edges from $\widetilde{G^{\frac{1}{2}}}$. Then, add a check node $c$ and the edges $v_1c, \ldots, v_{d_v-b+1}c$ to the graph. The resultant Tanner graph is a NETS in a variable-regular LDPC code with variable degree $d_v$ and girth at least $g$, in which all check nodes have degree two, except $b-1$ of them that have degree-$1$, and one that has the odd degree $d_v-b+1 \geq 3$. For the case of $b=0$ and $d_v$ even, remove the nodes $v,c_1, \ldots, c_{d_v}$, and their incident edges from $\widetilde{G^{\frac{1}{2}}}$. Then, add a check node $c$ and the edges $v_1c, \ldots, v_{d_v}c$ to the graph. The resultant Tanner graph is a NETS in a variable-regular LDPC code with variable degree $d_v$ and girth at least $g$, in which all check nodes have degree two, except $c$ that has an even degree $d_v \geq 4$ ($b=0$).

$(iii)(a)$ Let $G$ be a $(d_v;\frac{g }{2})$-graph. Consider the Tanner graph $\widetilde{G^{\frac{1}{2}}}$, and let $c$ be a check node in $\widetilde{G^{\frac{1}{2}}}$.  Add two variable nodes
$v_1,v_2 $, check nodes $c'_1, \ldots, c'_{d_v-1}, c''_1, \ldots, c''_{d_v-1}$, and edges $v_1c'_1, \ldots, v_1c'_{d_v-1}, v_2c''_1, \ldots, v_2c''_{d_v-1}, v_1c,v_2c $ to $\widetilde{G^{\frac{1}{2}}}$. The resultant Tanner graph is a NETS with $b= 2(d_v-1)$ in a variable-regular LDPC code with variable degree $d_v$ and girth $g$. (It has a check node of degree four and $2(d_v-1)$ degree-$1$ check nodes. The rest of the check nodes have degree two.)

$(iii)(b)$ Let $G$ be a $(d_v;\frac{g }{2})$-graph. Consider the Tanner graph $\widetilde{G^{\frac{1}{2}}}$, and let $c$ be a check node in $\widetilde{G^{\frac{1}{2}}}$. Add a  variable node  $v$, check nodes $c_1, \ldots, c_{d_v-1}$ and
the edges $vc_1, \ldots ,vc_{d_v-1}, vc $ to  $\widetilde{G^{\frac{1}{2}}}$. The resultant Tanner graph is a NETS with $b= d_v$ in a variable-regular LDPC code with variable degree $d_v$ and girth $g$.

$(iii)(c)$ Let $G$ be a $(d_v;\frac{g}{2})$-graph. Consider the Tanner graph $\widetilde{G^{\frac{1}{2}}}$. Let $v, v_1, \ldots, v_{d_v}$ be $d_v+1$ variable nodes  in $\widetilde{G^{\frac{1}{2}}}$, such that for each $1\leq i \leq d_v$, $v$ has a common neighbor $c_i$ with $v_i$. Remove the node $v$ and its incident edges from $\widetilde{G^{\frac{1}{2}}}$. Add variable nodes $v', v''$, check nodes  $c'_1, \ldots, c'_{d_v-1}, c''_1$, and edges $v'c'_1, \ldots, v'c'_{d_v-1}, v'c_1, v''c_1, v''c''_1, v''c_2, \ldots, v''c_{d_v-1}$ to the graph. The resultant Tanner graph is a NETS with $b=d_v+2$ in a variable-regular LDPC code with variable degree $d_v$ and girth at least $g$. (Note that $c'_1, \ldots, c'_{d_v-1}, c''_1, c_{d_v}$, all have degree one, $c_1$ has degree three, and the rest of the check nodes have degree two.)

$(iii)(d)$ Let $G$ be a $(d_v;\frac{g+2}{2})$-graph. Consider the Tanner graph $\widetilde{G^{\frac{1}{2}}}$. Let $v, v_1, \ldots, v_{d_v}$ be $d_v+1$ variable nodes  in $\widetilde{G^{\frac{1}{2}}}$, such that $v$ has a common neighbor $c_i$ with each node $v_i$, $i=1, \ldots, d_v$. Remove nodes $v, c_1, \ldots, c_{d_v}$ and their incident edges from $\widetilde{G^{\frac{1}{2}}}$. Then, add a variable node $v'$, check nodes $c'_1, \ldots, c'_{d_v}$ and the edges $v_1c'_1 , \ldots,v_{d_v}c'_1, v'c'_1, \ldots, v'c'_{d_v}$ to the graph.
The resultant Tanner graph is a NETS with $b=d_v-1$ in a variable-regular LDPC code with variable degree $d_v$ and girth at least $g$. (Note that $d(c'_1)= d_v+1$, which is even.)
}\end{proof}

\subsubsection{Upper bounds on $\mathcal{T}^{0}_{NETS}(d_v;g)$}

\begin{table}[ht]
\caption{Comparison of the upper bounds derived here with the lower bound of (\ref{ineq2}) on $\mathcal{T}^{0}_{NETS}(d_v;g)$}
\begin{center}
\scalebox{1}{
\begin{tabular}{ |l||c|c|c|c|c|c|   }
\hline
$\mathcal{T}^{0}_{NETS}(3;g)$  &  $g=$6  & $g=$8    & $g=$10    &$g=$12    & $g=$14     &$g=$16       \\
\hline
\hline
Lower bound of (\ref{ineq2})   &  8   &   12   &  18    &   28   &   40    &  60        \\
\hline
Derived upper bounds   &   8   &   12   &   18     &   28   &   40    &   60        \\
\hline
\multicolumn{7}{|l|}{$\mathcal{T}^{0}_{NETS}(4;g)$}       \\
\hline
Lower bound of (\ref{ineq2})   &  7   &   16   &  25    &   52   &   79    &  160        \\
\hline
Derived upper bounds   &   7   &   16   &   25     &   52   &   79    &   160        \\
\hline
\multicolumn{7}{|l|}{$\mathcal{T}^{0}_{NETS}(5;g)$}       \\
\hline
Lower bound of (\ref{ineq2})   &  8   &  20   &  36   &   84   &  148   &  340   \\
\hline
Derived upper bounds   &  12   &   20   &  60   &   84   &  304   &   340      \\
\hline
\multicolumn{7}{|l|}{$\mathcal{T}^{0}_{NETS}(6;g)$}       \\
\hline
Lower bound of (\ref{ineq2})   &  9   &  24   &  49   &   124   &  249   &  624   \\
\hline
Derived upper bounds   &  11   &  24   &  61   &   124   &  311   &   624      \\
\hline
\end{tabular}
}
\end{center}
\label{Table17E}
\end{table}

Fig.~\ref{F7} shows the construction used in Part $(i)(a)$ of Theorem~\ref{Th90} to derive the upper bound. The $(8,0)$ NETS in variable-regular LDPC codes with $d_v=3$ and $g=6$ is constructed starting from the $(3;3)$-cage shown in Fig.~\ref{F99}(a).

\begin{figure}[]
	\centering
	\includegraphics [width=0.3\textwidth]{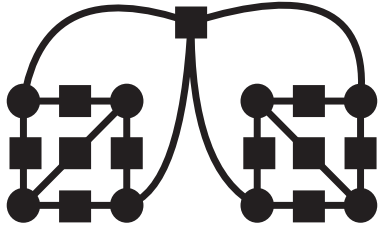}
	\caption{The $(8,0)$ NETS structure in variable-regular LDPC codes with $d_v=3$ and $g=6$, constructed from the $(3;3)$-cage shown in Fig.~\ref{F99}(a).}
	\label{F7}
\end{figure}

The upper bound derived in Part $(i)(a)$ of Theorem~\ref{Th90} appears to be tight for the cases where $g/2$ is an even number. In Table \ref{Table17E}, all the values for the cases where $g/2$ is even is derived based on this result.
As can be seen, in all cases, these upper bounds match the lower bound of (\ref{ineq2}).
For the cases where $\frac{g}{2}$ is an odd number, the upper bounds of Parts $(i)(b)$ and $(ii)$ (for $b=0$ and $d_v$ even) of Theorem~\ref{Th90} appear to be sometimes tighter than the upper bound of Part $(i)(a)$ of Theorem~\ref{Th90}. For the case of $d_v=3$, the following bound, whose proof is given in the appendix, is, in some cases, tighter than the bound derived in Part $(i)(b)$ of Theorem~\ref{Th90}.

\begin{fact}
Let $G$ be a $(3,4;\frac{g+2}{2})$-graph with $t \geq 1$ nodes of degree-$4$. Then, $\mathcal{T}^{0}_{NETS}(3;g)\leq |V(G)|+t-2$.
\label{thm64}
\end{fact}

Since $n'(3,4;5)=13$, using Fact~\ref{thm64}, we obtain the upper bound of $\mathcal{T}^{0}_{NETS}(3;8)\leq  12$, which is the same as the upper bound obtained from Part $(i)(a)$ of Theorem~\ref{Th90}.  For $\mathcal{T}^{0}_{NETS}(3;10)$, however, knowing that there is a $(3,4;6)$-graph with $18$ nodes, two of which have degree $4$~\cite{MR642638}, we have $\mathcal{T}^{0}_{NETS}(3;10)\leq 18$, which is tighter than the upper bound of $20$ from Part $(i)(a)$ of Theorem~\ref{Th90}. Also, there is a $(3,4;8)$-graph with $39$ nodes, three of which have degree 4~\cite{MR3127005}. This, based on Fact~\ref{thm64}, results in $\mathcal{T}^{0}_{NETS}(3;14)\leq 40$.

In Table \ref{Table17E}, we compare the upper bounds derived here for $\mathcal{T}^{0}_{NETS}(3;g)$ with the lower bounds of (\ref{ineq2}). As can be seen, all the derived upper bounds match the lower bounds and are thus tight.
(Note that for $\mathcal{T}^{0}_{NETS}(3;14)$, by Theorem \ref{lowf}, we obtain the lower bound of $39$, but since there is no Tanner graph with $39$ variable nodes such that the degree of each variable node is $3$ and the degree of each check node is even, the bound can be improved to $40$.)

For the case of $d_v=4$, the tightest bound (for odd values of $g/2$) is the one derived in Part $(ii)$ of Theorem~\ref{Th90}. In Table \ref{Table17E}, we have listed this upper bounds (for odd values of $g/2$) along with the lower bounds of (\ref{ineq2}). As can be seen, they all match.

In the case of $\mathcal{T}^{0}_{NETS}(5;g)$, the upper bounds listed in Table \ref{Table17E} are obtained by Part $(i)(a)$ of Theorem \ref{Th90}. As can be seen, while for $g=8, 12, 16$, the bounds are tight and meet the lower bounds, for the cases where $g/2$ is an odd number, there is a large gap between the two bounds.

For $d_v=6$, the tightest bound (for odd values of $g/2$) is the one derived in Part $(ii)$ of Theorem~\ref{Th90}.
This upper bound is listed for $\mathcal{T}^{0}_{NETS}(6;g)$ in Table \ref{Table17E} (for odd values of $g/2$).

\subsubsection{Upper bounds on $\mathcal{T}^{1}_{NETS}(d_v;g)$}

In Table \ref{Table18}, we have listed the upper bounds derived in Part $(ii)$ of Theorem~\ref{Th90} for $b=1$ and $d_v=3, 5$, in comparison with the lower bounds of Theorem~\ref{lowf}. Table  \ref{Table18} shows that while the two match for $\mathcal{T}^{1}_{NETS}(3;g)$ and some $g$ values, there is a gap between the two bounds for $\mathcal{T}^{1}_{NETS}(5;g)$.

\begin{table}[ht]
\caption{Comparison of the upper bounds of Theorem~\ref{Th90} with the lower bounds of Theorem~\ref{lowf} on $\mathcal{T}^{1}_{NETS}(d_v;g)$}
\begin{center}
\scalebox{1}{
\begin{tabular}{ |l||c|c|c|c|c|c|   }
\hline
$\mathcal{T}^{1}_{NETS}(3;g)$  &  $g=$6  & $g=$8    & $g=$10    &$g=$12    & $g=$14     &$g=$16       \\
\hline
\hline
Lower bound of Theorem~\ref{lowf}   &  5   &   9   &  13    &   21   &   29    &  45        \\
\hline
Upper bound of Theorem \ref{Th90}   &   5   &   9   &   13     &   23   &   29    &   57        \\
\hline
\multicolumn{7}{|l|}{$\mathcal{T}^{1}_{NETS}(5;g)$}       \\
\hline
Lower bound of Theorem~\ref{lowf}   &  7   &  15   &  31   &   63   &  127   &  255   \\
\hline
Upper bound of Theorem \ref{Th90}   &  9   &  29   &  41   &   151   &  169   &   --      \\
\hline
\end{tabular}
}
\end{center}
\label{Table18}
\end{table}

\subsubsection{Upper bounds on $\mathcal{T}^{2}_{NETS}(d_v;g)$}

To derive the upper bounds for $d_v = 3, 4$, and $6$, we use Parts $(iii)(d)$, $(ii)$, and $(ii)$ of Theorem~\ref{Th90}, respectively. For $d_v=5$, we use the follwing result whose proof is given in the appendix.

\begin{fact}
For any $g \geq 6$,  we have $\mathcal{T}^{2}_{NETS}(5;g)\leq 2n(5;\frac{g}{2})$.
\label{thmxx}
\end{fact}

In Table \ref{TableN99}, we compare the upper bounds derived here on $\mathcal{T}^{2}_{NETS}(3;g)$ with the corresponding lower bound of Theorem~\ref{lowf}. For the other values of $d_v$, the gap between the two bounds is rather  large, and thus the bounds are not listed in the table.

\begin{table}[ht]
\caption{Comparison of the upper bound of Theorem~\ref{Th90} and the lower bound of Theorem~\ref{lowf} for $\mathcal{T}^{2}_{NETS}(3;g)$}
\begin{center}
\scalebox{1}{
\begin{tabular}{ |l||c|c|c|c|c|c|   }
\hline
$\mathcal{T}^{2}_{NETS}(3;g)$  &  $g=$6  & $g=$8    & $g=$10    &$g=$12    & $g=$14     &$g=$16       \\
\hline
\hline
Lower bound of Theorem~\ref{lowf}  &  6   &   8   &  12    &   18   &   26    &  38        \\
\hline
Upper bound of Theorem~\ref{Th90}   &   6   &   10   &   14     &   24   &   30    &   58        \\
\hline
\end{tabular}
}
\end{center}
\label{TableN99}
\end{table}

\subsubsection{Upper bounds on $\mathcal{T}^{3}_{NETS}(d_v;g)$}

To derive the bounds for $d_v=3$ and $d_v=5$, we use Parts $(iii)(b)$ and $(ii)$ of Theorem~\ref{Th90}, respectively.
Table \ref{Table19} shows these upper bounds in comparison with the lower bounds of Theorem~\ref{lowf}. While the derived bounds match the lower bounds for $\mathcal{T}^{3}_{NETS}(3;g)$ and all $g$ values, except $g=14$, there is a gap between the upper and lower bounds for $\mathcal{T}^{3}_{NETS}(5;g)$.

\begin{table}[ht]
\caption{Comparison of the upper bounds of Theorem~\ref{Th90} and the lower bound of Theorem~\ref{lowf} for $\mathcal{T}^{3}_{NETS}(d_v;g)$}
\begin{center}
\scalebox{1}{
\begin{tabular}{ |l||c|c|c|c|c|c|   }
\hline
$\mathcal{T}^{3}_{NETS}(3;g)$  &  $g=$6  & $g=$8    & $g=$10    &$g=$12    & $g=$14     &$g=$16       \\
\hline
\hline
Lower bound  of Theorem~\ref{lowf}   &  5   &   7   &  11    &   15   &   23    &  31        \\
\hline
Upper bound of Theorem~\ref{Th90}   &   5   &   7   &   11     &   15   &   25    &   31        \\
\hline
\multicolumn{7}{|l|}{$\mathcal{T}^{3}_{NETS}(5;g)$}       \\
\hline
Lower bound  of Theorem~\ref{lowf}   &  7   &  13   &  29   &   53   &  117   &  213   \\
\hline
Upper bound of Theorem~\ref{Th90}    &  9   &  29   &  41   &   151   &  169   &   --      \\
\hline
\end{tabular}
}
\end{center}
\label{Table19}
\end{table}

\subsubsection{Upper bounds on $\mathcal{T}^{4}_{NETS}(d_v;g)$}

To derive the bounds for $d_v = 3, 4, 5$, and $6$, we use Parts $(iii)(a)$, $(iii)(b)$, $(iii)(d)$, and $(ii)$ of Theorem~\ref{Th90}. In Table \ref{Table20}, we compare these upper bounds with the lower bounds of Theorem \ref{lowf} for $\mathcal{T}^{4}_{NETS}(d_v;g)$, where $d_v=3,4$. For larger values of $d_v$,  the difference between the two bounds is large and thus the bounds are not presented in the table.

\begin{table}[ht]
\caption{Comparison of the upper bounds of Theorem~\ref{Th90} and the lower bounds of Theorem~\ref{lowf} for $\mathcal{T}^{4}_{NETS}(d_v;g)$}
\begin{center}
\scalebox{1}{
\begin{tabular}{ |l||c|c|c|c|c|c|   }
\hline
$\mathcal{T}^{4}_{NETS}(3;g)$  &  $g=$6  & $g=$8    & $g=$10    &$g=$12    & $g=$14     &$g=$16       \\
\hline
\hline
Lower bound of Theorem~\ref{lowf}   &  4   &   6   &  8    &   12   &   16    &  24        \\
\hline
Upper bound of Theorem~\ref{Th90}   &   6   &   8   &   12     &   16   &   26    &   32        \\
\hline
\multicolumn{7}{|l|}{$\mathcal{T}^{4}_{NETS}(4;g)$}       \\
\hline
Lower bound of Theorem~\ref{lowf}   &  5   &   9   &  15    &   27   &   45    &  81        \\
\hline
Upper bound of Theorem~\ref{Th90}   &   6   &   9   &   20     &   27   &   68    &   81        \\
\hline
\end{tabular}
}
\end{center}
\label{Table20}
\end{table}

\subsubsection{Upper bounds on $\mathcal{T}^{5}_{NETS}(d_v;g)$}

To derive the bounds for $d_v=3$ and $5$, we use Parts $(iii)(c)$ and $(iii)(b)$ of Theorem~\ref{Th90}, respectively.
These upper bounds are compared with the existing lower bound of Theorem \ref{lowf} in Table \ref{Table21}.

\begin{table}[ht]
\caption{Comparison of the upper bounds of Theorem~\ref{Th90} and the lower bounds of Theorem~\ref{lowf} for $\mathcal{T}^{5}_{NETS}(d_v;g)$}
\begin{center}
\scalebox{1}{
\begin{tabular}{ |l||c|c|c|c|c|c|   }
\hline
$\mathcal{T}^{5}_{NETS}(3;g)$  &  $g=$6  & $g=$8    & $g=$10    &$g=$12    & $g=$14     &$g=$16       \\
\hline
\hline
Lower bound of Theorem~\ref{lowf}   &  5   &   5   &  7    &   9   &   13    &  17        \\
\hline
Upper bound of Theorem~\ref{Th90}    &   5   &   7   &   11     &   15   &   25    &   31        \\
\hline
\multicolumn{7}{|l|}{$\mathcal{T}^{5}_{NETS}(5;g)$}       \\
\hline
Lower bound of Theorem~\ref{lowf}   &  7   &   11   &  23    &   43   &   91    &  171        \\
\hline
Upper bound of Theorem~\ref{Th90}   &   7   &   11   &   31     &   43   &   153    &   171        \\
\hline
\end{tabular}
}
\end{center}
\label{Table21}
\end{table}

\section{Conclusion}
\label{sec6}
In this paper, we established relationships between ``cages'' in graph theory, and ``trapping sets'' and ``codewords'' in coding theory. Based on these relationships, we used the extensive knowledge of cages in graph theory and
devised tight upper bounds on the size $a$ of the smallest $(a,b)$ LETSs, ETSLs, and NETSs, for different values of $b$. This was performed for variable-regular Tanner graphs with different variable degrees $d_v$ and girths $g$.
Of particular interest was the results on LETSs and NETSs for $b=0$, that correspond to upper bounds on the minimum weight of non-zero elementary and non-elementary codewords, respectively. For elementary codewords,
the derived bounds were proved to be always tight, i.e., they are equal to the exact minimum weight of non-zero elementary codewords. For non-elementary codewords, the derived upper bounds were tight for many values of $d_v$ and $g$, as they were shown to be equal to an existing lower bound. 

Finally, the connections between cages and trapping sets are helpful not only in the derivation of upper bounds on the size of smallest trapping sets, but also in learning about the non-isomorphic structures of such trapping sets. Since
trapping set structures of minimum size are often the most harmful ones, the results presented here may be used to analyze the error floor of LDPC codes or to design codes with low error floors.

\section{Appendix}

{\bf Proof of Fact~\ref{F1}:} Let $G$ be a $(5,6;\frac{g}{2})$-good-graph.\footnote{It can be shown that for any given value of $g \geq 6$, a  $(5,6;\frac{g}{2})$-good-graph exists. More generally, in the following, we show that for any integer $r \geq 2$, there always exists a $(2r-1,2r;g/2)$-good-graph, for any value of $g \geq 6$.  To see this, let $G$ be a $(2r-1;g/2)$-cage, and $e=xy$ be an arbitrary edge in $G$. Remove $e$ from $G$ and call the resulted graph $G'$.
Consider a simple cycle ${\cal C}$ of length $g/2$ with the set of nodes $v_1, v_2, \ldots, v_{g/2}$. For each value of $i$ in the range $2 \leq i \leq g/2$, add a node $u_i$ to ${\cal C}$, and connect it to the node $v_i$. 
For each $i$, $2 \leq i \leq g/2$, make $r-2$ copies of $G'$. Then, connect the node $v_i$ to both $x$ and $y$ in all copies of $G'$. Next, for each $i$, $2 \leq i \leq g/2$, make $r-1$ copies of $G'$, and connect the node $u_i$ to both $x$ and $y$ in all copies of $G'$. Finally, make $r-1$ copies of $G'$ and connect the node $v_1$, to both $x$ and $y$ in all copies of $G'$. It is easy to see that in the resultant graph all nodes have degree $2r-1$, except for the node $v_1$ that has degree $2r$, and that the girth of the graph is $g/2$.} Consider the Tanner graph $\widetilde{G^{\frac{1}{2}}}$. Let $v$ be the variable node of degree six and $c$ be a check node in $\widetilde{G^{\frac{1}{2}}}$ such that $vc \in E(\widetilde{G^{\frac{1}{2}}})$. Remove the edge $vc$ from $\widetilde{G^{\frac{1}{2}}}$. The resultant Tanner graph is a LETS with one unsatisfied check node in a variable-regular LDPC code with $d_v=5$ and girth at least $g$. \hfill
$ \blacksquare$

{\bf Proof of Fact~\ref{thm64}:} Let $G$ be a $(3,4;\frac{g+2}{2})$-graph. Assume that $G$ has $t \geq 1$ nodes of degree-4. For each degree-4 node (except one of them, say $w$), perform the following procedure: Let $v$ be the node of degree $4$ and $u_1,u_2,u_3,u_4$ be its neighbors. Remove the node $v$ and its incident edges from $G$, and then add two nodes $v_1,v_2$ and the edges $v_1v_2, v_1u_1,v_1u_2,v_2u_3,v_2u_4$ to the graph. By having done this procedure for all  degree-$4$ nodes except $w$, we obtain a graph such that the degree of each node except $w$ is three and the girth is at least $\frac{g+2}{2}$. Call this graph $H$. Consider the Tanner graph $\widetilde{H^{\frac{1}{2}}}$.
Assume that $\{c_1,c_2,c_3,c_4\}$ is the set of neighbors of $w$ and for each $i$, $1\leq i \leq 4$, there is variable node $w_i $, such that $w_ic_i \in E(\widetilde{H^{\frac{1}{2}}})$.
Remove the nodes $w,c_1,c_2,c_3,c_4$ and their incident edges form the graph. Next, add a check node $c'$ and the edges $c'w_1,c'w_2,c'w_3,c'w_4$ to the graph.
The resultant graph is a NETS structure with $b=0$ in a variable-regular LDPC code with $d_v=3$ and girth at least $g$. \hfill $ \blacksquare$

{\bf Proof of Fact~\ref{thmxx}:} Let $G$ be a $(5;\frac{g }{2})$-graph. Consider the Tanner graph $\widetilde{G^{\frac{1}{2}}}$. Let $v_1,v_2$ be two variable nodes in $\widetilde{G^{\frac{1}{2}}}$, and  $c $ be  a check node such that $cv_1,cv_2 \in E(\widetilde{G^{\frac{1}{2}}})$. Consider the union of two copies of the graph  $\widetilde{G^{\frac{1}{2}}}$. Remove the node $c$ and its incident edges from both copies of  $\widetilde{G^{\frac{1}{2}}}$. Then, add a check node $c'$ and connect it to variable nodes  $v_1,v_2$ in the first copy and the second copy of $\widetilde{G^{\frac{1}{2}}}$. In the resultant graph, let $z$ be a degree-$2$ check node with the set of neighbors $\{u_1,u_2\}$. Remove $z$ and its incident edges from the graph. Add two check nodes $z',z''$ and the edges $u_1z',u_2z''$ to the graph. The resultant Tanner graph is a NETS  with $b=2$  in a variable-regular LDPC code with $d_v=5$ and girth at least $g$. \hfill $\blacksquare$

\bibliographystyle{ieeetr}

\end{document}